\documentclass{sn-jnl}% Math and Physical Sciences Numbered Reference Style 
\usepackage{graphicx} %
\usepackage{mathrsfs}
\usepackage{enumitem}
\usepackage{amsthm}
\usepackage[super]{nth}
\usepackage{mathtools}
\usepackage{color}
\usepackage{tabularx}
\usepackage{amsmath}
\usepackage{csquotes}
\usepackage{multicol}
\usepackage{amssymb}
\usepackage[T1]{fontenc}
\usepackage{fancyvrb}
\usepackage{makecell}
\usepackage{dsfont}
\usepackage{caption}
\usepackage{subcaption}
\usepackage{mdframed}
\usepackage{empheq}
\usepackage{arydshln}
\usepackage{nicefrac}
\usepackage{dsfont}
\usepackage{wrapfig}
\usepackage{wrapstuff}
\usepackage{etoolbox}
\makeatletter
\patchcmd\WF@putfigmaybe{\lower\intextsep}{}{}{\fail}%
\AddToHook{env/wrapfigure/begin}{\setlength{\intextsep}{0pt}}
\makeatother
\usepackage{natbib}
\usepackage{cleveref}
\usepackage{pgfplots}
\pgfplotsset{compat=1.18}
\usepackage[ruled, linesnumbered]{algorithm2e} 

\usepackage{array}
\newcolumntype{L}[1]{>{\raggedright\let\newline\\\arraybackslash\hspace{0pt}}m{#1}}
\newcolumntype{C}[1]{>{\centering\let\newline\\\arraybackslash\hspace{0pt}}m{#1}}
\newcolumntype{R}[1]{>{\raggedleft\let\newline\\\arraybackslash\hspace{0pt}}m{#1}}

\usepackage{tikz}
\usetikzlibrary{automata, positioning, calc, shapes, arrows, fit, patterns}

\DeclareMathOperator{\rank}{rank}

\DeclareMathOperator{\supp}{supp}

\theoremstyle{plain}
\newtheorem{theorem}{Theorem}

\newtheorem{example}{Example}

\newtheorem{definition}{Definition}

\newtheorem{observation}[theorem]{Observation}
\usepackage{soul}
\usepackage{etoolbox}
\makeatletter
\patchcmd{\SOUL@ulunderline}{\dimen@}{\SOUL@dimen}{}{}
\patchcmd{\SOUL@ulunderline}{\dimen@}{\SOUL@dimen}{}{}
\patchcmd{\SOUL@ulunderline}{\dimen@}{\SOUL@dimen}{}{}
\newdimen\SOUL@dimen
\makeatother
\sethlcolor{blue!20}
\soulregister\cite7
\soulregister\ref7
\soulregister\Cref7
\soulregister\citet7
\soulregister\citep7
\soulregister\footnote8
\newcommand{\new}[1]{#1}
\date{\today}

\begin{document}

\title[Candidate Monotonicity and Proportionality for Lotteries and Non-Resolute Rules]{Candidate Monotonicity and Proportionality for Lotteries and Non-Resolute Rules}


\author[1]{\fnm{Jannik} \sur{Peters}}\email{peters@nus.edu.sg}

\affil[1]{\orgdiv{School of Computing}, \orgname{National University of Singapore}, \orgaddress{\country{Singapore}}}


\abstract{We study the problem of designing multiwinner voting rules that are candidate monotone and proportional. We show that the set of committees satisfying the proportionality axiom of proportionality for solid coalitions is candidate monotone. We further show that Phragmén's Ordered Rule can be turned into a candidate monotone probabilistic rule which randomizes over committees satisfying proportionality for solid coalitions.}




\maketitle

\section{Proportionality and Monotonicity}
Proportional multiwinner voting rules are among the most well studied in computational social choice. Among them, the \emph{Single Transferable Vote (STV)} is used in several countries, e.g., Australia, Malta, or Scotland, to elect governing bodies. One of the main reasons STV can be considered proportional is its satisfaction of the \emph{proportionality for solid coalitions (PSC)} axiom \citep{Dumm84a, Tide95a}. In essence, for ordinal preferences lists PSC requires that every group of voters listing the same block of candidates (in potentially different permutations) at the start of their preference list, should get proportionally many candidates from this block. For instance, if $5$ candidates were to be elected and $40\%$ of the voters have the same three candidates in their top three, then two out of these three candidates need to be chosen in the end.\footnote{This is slightly simplifying things, the commonly used variants of STV usually work with a ``quota'' that is slightly lower, the so-called Droop quota $\lceil \frac{n}{k + 1}\rceil$ with $n$ being the number of voters and $k$ the number of candidates to be elected. For ease of exposition, we focus on the Hare quota $\frac{n}{k}$ instead. The results also hold for the Droop quota.} One of the major flaws of STV is its failure of \emph{candidate monotonicity}: strengthening a winning candidate can make this candidate stop winning \citep{Wood97a, McGr23a}. Interestingly enough, whether this flaw is a necessity is an open question. Namely, it is unknown whether there is any voting rule satisfying PSC together with candidate monotonicity \citep{Wood97a, AzLe20a} with most known proportional voting rules failing candidate monotonicity already for very small \new{instance sizes} \citep{AzLe20a}. Despite \citeauthor{Wood97a} already writing that ``a major unsolved problem is whether there exist rules that retain the important political features of STV and are also more monotonic'' \citep[page~1]{Wood97a} no such rule has been found or proven to be nonexistent.

\subsubsection*{Our Contribution} In this paper, we tackle this open question from two slightly weaker angles. Firstly, we study \emph{non-resolute} voting rules: rules that do not just output one committee, but a whole set of committees. For this setting, we show that the rule outputting all committees satisfying PSC is indeed candidate monotone. Secondly, we study \emph{probabilistic} voting rules, i.e., rules selecting a lottery over committees. For this, we show that a variant of Phragmén's Ordered Rule \citep{Jans16a} can be adapted into a candidate monotone lottery over committees satisfying PSC. 

\subsubsection*{Related Work}
Our work is situated in the area of proportional multiwinner voting with ranked preferences, in particular in the axiomatic study of multiwinner voting rules. This study dates back to the work of eminent philosopher \citet{Dumm84a} and his development of PSC. \citet{AzLe20a} were the first to give slightly improved results for monotonicity and proportionality, showing that their Expanding Approvals Rule (EAR) satisfies a weaker form of monotonicity together with a generalization of PSC in the setting of weak ranked preferences in which ties are allowed. The same generalization of PSC together with a form of monotonicity they call ``indifference monotonicity'' is also satisfied by the generalization of STV by \citet{DePe24a}. \citet{McGr23a} experimentally studied monotonicity violations in a real-world dataset consisting of Scottish local elections and found that out of 1079 elections 62 possessed some kind of ``monotonicity anomaly''. In particular, in $23$ elections STV failed the version of candidate monotonicity studied by us.

PSC was further investigated by \citet{BrPe23a} who criticized it for being a relatively weak axiom and introduced stronger proportionality axioms still satisfied by EAR. Further, \citet{AzLe21a} introduced PSC generalizations for the context of participatory budgeting \citep{PPS21a} and in \citet{AzLe22a} provided a characterization of PSC in the standard ordinal preference setting. Finally, recently \citet{ALR24a} studied committee monotonicity, i.e., the selection of a ranking of candidates, together with PSC and showed that a rule satisfying both PSC and committee monotonicity exists.

In the setting of approval-based multiwinner voting, monotonicity is significantly easier to obtain \citep{SaFi17a} with most rules being monotone in the sense that voters additionally approving an already winning candidate cannot hurt that candidate.\footnote{In fact, all voting rules studied in the recent book of \citeauthor{LaSk22a} on approval-based multiwinner voting satisfy candidate monotonicity \citep[Table~3.1]{LaSk22a}} However, it can for instance happen that adding voters only approving a single candidate can make this candidate not be chosen anymore (see \citet{BFJ+24a} for examples of this in the context of participatory budgeting). 

Finally, our results are also situated in the area of probabilistic social choice. Here, we are closely related to recent works focusing on \emph{best of both world} guarantees, i.e., the design of lotteries satisfying desirable properties both ex-ante and ex-post. In particular, recent work focusing on the design of randomized apportionment schemes \citep{GPP24a, JGU+24a} and on achieving ex-post and ex-ante proportionality guarantees in approval-based multiwinner voting and participatory budgeting \citep{ALS+23a, ALS+24a, SuVo24a} are close to our section on probabilistic voting rules. 

\section{Preliminaries}

\subsubsection*{Notation}
For $t \in \mathbb{N}$ let $[t]\coloneqq \{1, \dots, t\}$. We are given a set $N = [n]$ of \emph{voters} and a set $C = \{c_1, \dots, c_m\}$ of \emph{candidates}. Each voter $i \in N$ possesses a strict linear ranking $\succ_i \subseteq C \times C$ over all candidates. Further, we are given a \emph{committee size} $k \le m$. Together, the voters, candidates, preferences, and committee size form an instance $\mathcal{I} = (N, C, (\succ_i)_{i \in N}, k)$. A \emph{committee} is a subset $W \subseteq C$ of size $\lvert W\rvert = k$. Let $C_k = \{W \subseteq C \colon \lvert W \rvert = k\}$ be the set of all committees. A \emph{rule} $r$ is a function mapping an instance $\mathcal{I}$ to a non-empty set of committees $r(\mathcal{I}) \subseteq C_k$. If $\lvert r(\mathcal{I})\rvert = 1$ for all instances $\mathcal{I}$ the rule is said to be resolute.

\sloppy For probabilistic rules, we mostly follow the notation of \citet{ALS+23a}. We say that a fractional committee is an assignment of probabilities ${(p_1, \dots, p_m) \in [0,1]^m}$ for each candidate with $\sum_{i \in [m]} p_i = k$ with a fractional committee rule mapping instances to fractional committees. We let $\Delta(C_k)$ be the set of all probability distributions over $C_k$. A probabilistic rule $r$ then maps an instance $\mathcal{I}$ to a probability distribution $r(\mathcal{I}) \in \Delta(C_k)$. Each probabilistic rule also has a natural non-resolute counterpart $\supp(r(\mathcal{I}))$, i.e., the set of committees being selected with non-zero probability (or the support of $r(\mathcal{I})$). Further, each probabilistic rule $r$ implements a fractional committee voting rule, with the probability of each candidate being the probability of this candidate being on the committee selected by $r$. On the other hand, there can be several probabilistic voting rules implementing a fractional committee voting rule. 
\subsubsection*{Proportionality}
 Following the seminal work of \citet{Dumm84a} we define the axiom of \emph{proportionality for solid coalitions (PSC)}. For this, we first need to define the eponymous solid coalitions. 
\begin{definition}[Solid Coalition]
    A group $N' \subseteq N$ of voters forms a \emph{solid coalition} over a group $C' \subseteq C$ of candidates if 
    \[
    c' \succ_i c \text{  for all } i \in N', c' \in C', \text{ and } c \in C\setminus C'.
    \]
\end{definition}
In essence, PSC now requires that every solid coalition is represented proportionally. 
\begin{definition}[Proportionality for Solid Coalitions \citep{Dumm84a}]
    A committee $W$ satisfies \emph{proportionality for solid coalitions (PSC)} if there is no group $N' \subseteq N$ of voters, $C' \subseteq C$ of candidates and $\ell \in [k]$ such that 
    \begin{itemize}
        \item[i.)] $N'$ forms a solid coalition over $C'$
        \item[ii.)] $\lvert N'\rvert \ge \ell \frac{n}{k}$
        \item[iii.)] $\lvert W \cap C'\rvert < \min(\lvert C'\rvert, \ell)$. 
    \end{itemize}
\end{definition}
In other words, if $C'$ forms a prefix of the preferences of all voters in $N'$ then $W$ needs to contain a large enough ``proportional share'' of them to satisfy $N'$. 

There are several rules satisfying PSC, among them the famous \emph{Single Transferable Vote (STV)} \citep{Tide95a} and the \emph{Expanding Approvals Rule (EAR)} \citep{AzLe20a}. 
Note that due to our definition of rules the set of all committees satisfying PSC is also a (non-resolute) rule.

\begin{example}
    Consider the following instance with four voters, four candidates, and $k = 2$:
    \begin{itemize}
        \item $c_1 \succ c_2 \succ c_3 \succ c_4$
        \item $c_2 \succ c_1 \succ c_3 \succ c_4$
        \item $c_3 \succ c_4 \succ c_1 \succ c_2$
        \item $c_4 \succ c_2 \succ c_1 \succ c_3$.
    \end{itemize}
    This instance contains one non-trivial solid coalition large enough to be able to demand a candidate, namely the first and second voter over $\{c_1, c_2\}$ (and technically also over $\{c_1, c_2, c_3\}$). Thus, any committee containing at least one of $c_1$ or $c_2$ satisfies PSC. Now if the third voter would swap candidate $c_4$ up, the third and fourth voter would form a solid coalition over $\{c_4\}$. Therefore, in this new instance the only two PSC committees are $\{c_1, c_4\}$ and $\{c_2, c_4\}$
\end{example}
\vspace{-0.7em}
\subsubsection*{Monotonicity} For our definition of (candidate) monotonicity for non-resolute voting rules, we follow \citet{EFSS17a}.
\begin{definition}[Candidate Monotonicity \citep{EFSS17a}]
    A rule $r$ satisfies \emph{candidate monotonicity} if for every instance $\mathcal{I}$, committee $W \in r(\mathcal{I})$ and $c \in W$ in any instance $\mathcal{I}'$ obtained by shifting $c$ one rank up in the preference list of any voter, there is still some committee $W' \in r(\mathcal{I}')$ with $c \in W'$.
\end{definition}
That is, if the candidate $c$ is in some winning committee and is ``made stronger'' $c$ is still in some winning committee afterward. It is known that neither STV nor EAR are candidate monotone \citep{AzLe20a}. However, Phragm\'en's ordered rule is indeed candidate monotone \citep{Jans16a}. 

To generalize candidate monotonicity to probabilistic rules, we instead focus on the probabilities of individual candidates. 
\begin{definition}[Candidate Monotonicity for Probabilistic Rules]
    A probabilistic rule $r$ satisfies candidate monotonicity if for every instance $\mathcal{I}$ and $c \in C$ in any instance $\mathcal{I}'$ obtained by shifting $c$ one rank up in the preference list of any voter it holds that 
    \[
    \Pr[c \in r(\mathcal{I})] \le \Pr[c \in r(\mathcal{I}') ].
    \]
\end{definition}
In other words, shifting a candidate up in the preference lists of voters does not decrease their chances to be selected. This is indeed in some sense a stronger requirement than candidate monotonicity for non-resolute rules.
\begin{observation}
    If $r$ is a candidate monotone probabilistic rule, then $\supp(r)$ is a candidate monotone non-resolute rule.
\end{observation}
Similarly, one can also define candidate monotonicity for fractional voting rules by looking at the individual probabilities.

\section{Non-Resoluteness and Monotonicity}
We start off our results by showing that the rule that selects all PSC committees is candidate monotone. 
For this we crucially make use of a characterization of PSC by \citet{AzLe22a} using their \emph{minimal demand rules}. These works as follows:
We start off with an empty committee $W = \emptyset$. For each rank $r = 1, \dots, m$ we check if there is a group of voters forming a solid coalition over a set $C'$ of candidates of size $r$. If, additionally, this solid coalition witnesses a PSC violation, we add a candidate from $C' \setminus W$ to $W$. We repeat this until there is no further PSC violation, afterwards we increment $r$. A committee now satisfies PSC if and only if it can be obtained via this process.
\begin{theorem}[Characterization of PSC \citep{AzLe22a}]
    A committee $W$ satisfies PSC if and only if it can be obtained by a minimal demand rule.
\end{theorem}
Using this characterization, we can now show that the set of PSC committees is indeed candidate monotone.
\begin{theorem}
    The rule selecting all PSC committees is candidate monotone.
\end{theorem}
\begin{proof}
    \new{Let $\mathcal{I}$ be an instance, $W$ a committee satisfying PSC in $\mathcal{I}$, and $c \in W$. Assume that there is a voter swapping $c$ up in their preference list and denote this instance by $\mathcal{I}_c$. We now construct a new committee $W_c$ containing $c$ by running a minimal demand rule on $\mathcal{I}_c$. We start off by initializing $W_c = \emptyset$ and iteratively check for each rank $r = 1, \dots, m$ if there exists a solid coalition consisting of a set $N' \subseteq N$ of voters and $C' \subseteq C$ of candidates witnessing a PSC violation in $\mathcal{I}_c$. If $c$ is already a part of $W_c$ we can include an arbitrary candidate from $C'\setminus W_c$ to $W_c$. Otherwise, we distinguish two cases. 

    \textbf{Case 1: $c \notin C'$.} If this is the case, as $c$ was swapped up, the set $N'$ together with $C'$ already formed a solid coalition in the original instance $\mathcal{I}$. Thus, as $W$ satisfies PSC there must be more candidates in $W \cap C'$ than in $W_c \cap C'$. Hence, there exists a candidate $c' \in (W \cap C') \setminus W_c$ and we can include this candidate in the committee $W_c$.

    \textbf{Case 2: $c \in C'$.} Then we include $c$ in the committee $W_c$.

    This constitutes a valid execution of a minimal demand rule. Thus, the committee $W_c$ satisfies PSC. If we ever entered the second case, we know that $c \in W_c$ has to hold. Otherwise, we always entered the first case. However, then the $k$ candidates selected by $W_c$ are precisely the $k$ candidates in $W$. As $c \in W$ this implies $c \in W_c$.} 
\end{proof}

% \end{proof}
\section{A Probabilistic Version of Phragmén's Ordered Rule}
While the previous result gives us some hope for obtaining a proportional and candidate monotone rule, \new{the result} is also quite weak. There can be a lot of committees satisfying PSC \citep{BBMP25a} and our monotonicity notion only guarantees us that some committee includes the strengthened candidate. To remedy this, we study probabilistic voting rules. In particular, we use a variant of \emph{Phragmén's Ordered Rule} \citep{Jans16a} and turn it into a lottery over committees satisfying PSC.

\subsubsection*{Phragmén's Ordered Rule} Phragmén's Ordered Rule was originally introduced by a Swedish royal commission (including Phragmén himself) and was used for the intra party seat distribution in the Swedish parliament. \new{For} more historic context see the work by \citet{Jans16a}. 

Phragmén's Ordered Rule works as a continuous process. Each voter starts to ``eat'' their top choice candidate at the same speed. Once a candidate is fully eaten, this candidate is elected onto the committee and the voters move on to their top choice ``uneaten'' candidate. This is done until $k$ candidates are elected.\footnote{This is very similar to the famous probabilistic serial method of \citet{BoMo01a} and the veto by consumption rule of \citet{IaKo21a}. For a connection between these three methods and the problem of designing a voting rule with low metric distortion see the note by \citet{Pete23a}.}

\begin{example}
    Consider the following instance with four voters, four candidates, and $k = 2$:
\begin{itemize}
    \item $c_1 \succ c_4 \succ c_3 \succ c_2$
    \item $c_1 \succ c_3 \succ c_2 \succ c_4$
    \item $c_2 \succ c_3 \succ c_4 \succ c_1$
    \item $c_3 \succ c_4 \succ c_2 \succ c_1$
\end{itemize}

In this instance Phragm\'en's Ordered Rule works as follows:

\begin{enumerate}
    \item At time $t = 0$, all voters start ``eating'' their top choice candidate.

    \item At time $t = 0.5$:
    \begin{itemize}
        \item $c_1$ is fully eaten and elected to the committee.
        \item $c_2$ and $c_3$ are ``half eaten''.
        \item The first two voters move to their next choices $c_4$ and $c_3$.
        \item The current committee is $\{c_1\}$.
    \end{itemize}

    \item At time $t = 0.75$:
    \begin{itemize}
        \item $c_3$ is now fully eaten by the second and fourth voter, \new{$c_2$} is $75\%$ eaten and $c_4$ is $25\%$ eaten. 
        \item $c_3$ is elected to the committee and $\{c_1, c_3\}$ is the outcome.
    \end{itemize}
\end{enumerate}
\end{example}
To see that Phragmén's Ordered Rule does not satisfy PSC consider the following example.
\begin{example}
    Consider the following instance with $15$ voters, $8$ candidates and $k = 3$:
    \begin{itemize}
        \item[$2\times$] $c_1\succ c_2 \succ c_3 \succ c_4 \succ c_5 \succ c_6 \succ c_7$
        \item[$2\times$] $c_2\succ c_3 \succ c_4 \succ c_5 \succ c_1  \succ c_6 \succ c_7$
        \item[$2\times$] $c_3\succ c_4 \succ c_5  \succ c_1 \succ c_2 \succ c_6 \succ c_7$
        \item[$2\times$] $c_4 \succ c_5\succ c_1 \succ c_2 \succ c_3  \succ c_6 \succ c_7$
        \item[$2\times$] $c_5\succ c_1 \succ c_2 \succ c_3  \succ c_4 \succ  c_6 \succ c_7$
        \item[$5 \times$] $c_6 \succ c_7 \succ c_1\succ c_2 \succ c_3 \succ c_4 \succ c_5$
    \end{itemize}
    Here, Phragmén's Ordered Rule would first elect $c_6$. Then at the time $c_7$ is fully eaten, every candidate from $c_1$ to $c_5$ is $80\%$ eaten. Thus, Phragmén's Ordered Rule first elects $c_6$ and $c_7$ and only one of $c_1$ to $c_5$ afterwards. However, the first ten voters form a solid coalition over $c_1$ to $c_5$ demanding at least two candidates while they only receive one. Therefore, Phragmén's Ordered Rule fails PSC in this instance. \label{ex:phrag:nopsc}
\end{example}

We now turn this deterministic rule first into a fractional one and then via dependent rounding into a probabilistic one. Specifically, we first define a version of Phragmén's Ordered Rule giving us probabilities for the individual candidates to be elected and then show how to round these probabilities to guarantee that every committee in the support satisfies PSC.  

We assume that each candidate has a weight of one. Now instead of stopping when $k$ candidates have been selected, we stop once the total eaten weight is $k$. Candidates can now be eaten fractionally, with $\alpha$ of a candidate being eaten corresponding to an $\alpha$ probability of this candidate being selected. We call this rule Phragm\'en's fractional rule (PFR).

\begin{example}
    Consider the following instance with four voters, four candidates, and $k = 2$:
\begin{itemize}
    \item $c_1 \succ c_2 \succ c_3 \succ c_4$
    \item $c_1 \succ c_4 \succ c_3 \succ c_2$
    \item $c_1 \succ c_3 \succ c_2 \succ c_4$
    \item $c_3 \succ c_4 \succ c_2 \succ c_1$
\end{itemize}    
After the total eaten weight is $\frac{4}{3}$, candidate $c_1$ is eaten fully, while candidate $c_3$ is eaten a third. Then the first three voters move on to their second choice. After the total eaten weight is $2$ candidates $c_2$ and $c_4$ both have a weight of $\frac{1}{6}$, while $c_3$ has $\frac{2}{3}$. Thus, candidate $c_1$ is selected with probability $1$ while $c_2$ and $c_4$ are selected with probability $\frac{1}{6}$ and $c_3$ is selected with probability $\frac{2}{3}$. 

In the instance in \Cref{ex:phrag:nopsc}, $c_6$ would get a probability of $1$ while $c_1$ to $c_5$ each get a probability of $\frac{2}{5}$. 
\end{example}
While PFR gives us probabilities, it does not actually give us a distribution over committees yet. For this, we turn to the dependent rounding scheme of \citet{GKPS06a}. 
\begin{theorem}[\citep{GKPS06a}]
    Let $G = (V,E)$ be a bipartite graph with weights $w_e \in [0,1]$ for each $e \in E$. Then there exists a lottery over subsets of edges $(\lambda_i, E_i)_{E_i \subseteq E}$ where the subset of edges $E_i \subseteq E$ is selected with probability $\lambda_i$ such that (i) each edge is selected with probability $w_e$ and (ii) in each $E_i$ a vertex $v$ has at least a degree of $\lfloor \sum_{e \in N(v)} w(e)\rfloor$ and at most $\lceil \sum_{e \in N(v)} w(e)\rceil$.
    \label{thm:dependent}
\end{theorem}
That is given a bipartite graph with probabilities on edges, we can round these edges such that the original probabilities are respected and the degrees in the selected graph are the rounded degrees in the original weighted graph.

Before we turn to the rounding scheme, we additionally introduce another property satisfied by PFR. Namely, it satisfies an ex-ante version of PSC, which requires that any group of at least $\frac{\ell n}{k}$ voters (even for non-integer $\ell$) gets a probability mass of at least $\ell$. 
\begin{definition}[Ex-ante PSC]
    A fractional committee rule $p$ satisfies ex-ante PSC if there is no group $N' \subseteq N$ of voters, $C' \subseteq C$ of candidates and $\ell \le k$ such that 
    \begin{itemize}
        \item[i.)] $N'$ forms a solid coalition over $C'$
        \item[ii.)] $\lvert N'\rvert \ge \ell \frac{n}{k}$
        \item[iii.)] $\sum_{c \in C'} p_c < \min(\lvert C'\rvert, \ell)$. 
    \end{itemize}
\end{definition}
\new{To better contexualize ex-ante PSC, we note that a (deterministic) committee satisfying ex-post PSC does not necessarily need to satisfy ex-ante PSC.\footnote{As a simple example consider an instance with two voters, one ranking $c_1 \succ c_2$ and the other ranking $c_2 \succ c_1$. Here, for $k = 1$, either $\{c_1\}$ or $\{c_2\}$ satisfies ex-post PSC. However, to satisfy ex-ante PSC, both $c_1$ and $c_2$ need to be selected with probability $\frac{1}{2}$.}  }

With this definition we now come to the main result of this section and show that PFR is indeed candidate monotone and satisfies this ex-ante PSC notion. Further, we use the dependent rounding scheme of \citet{GKPS06a} to show that the probabilities given by PFR can be decomposed into a probabilistic rule only randomizing over committees satisfying PSC.
\begin{theorem}
    PFR is candidate monotone, satisfies ex-ante PSC, and can be decomposed into a distribution over PSC committees. \label{thm:main}
\end{theorem}
\begin{proof}
    The monotonicity of PFR follows from the proof of \citet[Theorem~14.3]{Jans16a}, as he shows that at any time, after a candidate gets swapped up, this candidate is always ``at least as much eaten'' as before swapping this candidate up. In particular this holds for time step $k$ and thus the monotonicity follows.  

    For the ex-ante PSC guarantee, consider any group $N' \subseteq N$ of voters of  size at least $\ell \frac{n}{k}$ forming a solid coalition over a set $C'\subseteq C$ of candidates. After a total weight of $k$ has been eaten, the weight of the candidates eaten by the group $N'$ is exactly $\ell$. As $C'$ forms a prefix of the preferences of these voters, it must hold that either $C'$ is eaten in total or that at least $\ell$ probability mass is spent on candidates in $C'$ as voters in $N'$ would first start eating candidates outside of $C'$ once everything in $C'$ is eaten. Thus, the outcome satisfies ex-ante PSC.

    For the ex-post guarantee, we make use of the dependent rounding scheme of \Cref{thm:dependent} while following the general structure of the minimal demand rules of \citet{AzLe22a}. 
    
    Let $p$ be the fractional committee selected by PFR. For a candidate $c \in C$ and voter $i \in N$ let $p_c(i) \in [0,1]$ be the amount of $c$ eaten by $i$ during the execution of PFR.

    We start off with the bipartite graph $G = (\emptyset \cup C, E)$ \new{which we will later use to round the outcome}. As a first step, for each candidate $c \in C$ with $p_c = 1$ we add a voter $v_c$ to $G$ with an edge to $c$ of weight $1$ and set $p_c(i) = 0$ for all $i \in N$. \new{Note that if we round $G$ according to the dependent rounding scheme, this forces the edge from $v_c$ to $c$ to always be selected.}  Now, we go from rank $r = 1$ to $m$. If there is a group $N' \subseteq N$ of voters forming a solid coalition with $C'$ of size $\lvert C'\rvert = r$, let $\ell = \left\lfloor \sum_{i \in N'} \sum_{c \in C'} p_c(i) \right\rfloor$ \new{be the rounded ``budget'' left by this group of voters}. If $\ell = 0$ we continue. Otherwise, for each $c \in C'$ and $i \in N'$ \new{we select a value} $p'_c(i) \in [0,1]$  such that $p'_c(i) \le p_c(i)$ and $\sum_{c \in C', i \in N'} p'_c(i) = \ell$. We add a group voter $i'$ to the left side of the bipartite graph and connect it to each candidate in $c \in C'$ with an edge of weight $\sum_{i \in N'} p'_c(i)$. Further, we deduct each $p_c(i)$ by $p'_c(i)$. 
 If there is no group left, we increase $r$.

 The dependent rounding scheme now selects committees according to \new{the bipartite graph} $G$. \new{That is, we round the edges and select every candidate adjacent to an edge in the rounded outcome. As the weight adjacent to each group voter is integral, and as the total sum of weights is $k$ this guarantees that exactly $k$ candidates are selected. As the weight adjacent to each candidate is at most $1$, we also do not select any candidate twice.}

 Besides this, we need to show two things, (i) this rounding scheme respects the probabilities and (ii) the committees output by the rounding scheme satisfy PSC. 

 \new{For the first point, consider any candidate $c \in C$ and the vertices of $G$
 connected to $c$. These vertices correspond to solid coalitions selected throughout the construction of $G$. By construction, for each voter in these solid coalitions, the probability that $c$ is selected is increased by at most $p_c(i)$. Hence, as the dependent rounding scheme respects the probabilities on the edges, $c$ is selected with a probability of at most $\sum_{i \in N} p_c(i) = p_c$.
 
 Further, we know that for $r = m$ all voters form a solid coalition over all candidates. If at this step, there is any voter $i \in N$ left with $p_c(i) > 0$, this probability will be part of the last vertex added to $G$. Thus, we know that $c$ is selected with a probability of exactly $p_c$.} 

 For (ii) consider a group $N' \subseteq N$ of voters forming a solid coalition over a set of candidates $C' \subseteq C$ such that $\lvert N'\rvert \ge \frac{\ell n}{k}$ for some $\ell \in [k]$. \new{Our goal is to show that at least $\min(\ell, \lvert C'\rvert)$ candidates are selected from $C'$.} Consider the step in PFR in which $r = \lvert C' \rvert$. First, we notice that any candidate from $C'$ who has a probability of $1$ of getting selected is already included in the committee. If this includes every candidate in $C'$ we are already done, as in this case, every candidate from $C'$ is selected independent of the rounding. 
 
 \new{Otherwise, let $C''$ be the set of candidates from $C'$ with a probability of less than $1$ of being selected. As the voters in $N'$ rank all candidates in $C'$ before all candidates in $C \setminus C'$ we know that no voter in $N'$ can eat any candidate in $C \setminus C'$ before all candidates in $C'$ are assigned a probability of $1$. As we ran PFR until the total eaten weight was $k$, we get that the voters in $N'$ consumed at least $\lvert N'\rvert \frac{k}{n} \ge \ell$ many fractional candidates from $C'$. For each rank in-between $1$ and $\lvert C'\rvert$ voters in $N'$ could have only taken part in solid coalitions containing exclusively candidates from $C'$. Now consider $p' \coloneqq \sum_{i \in N'}\sum_{c \in C'} p_c(i)$: the ``budget'' left at iteration $\lvert C'\rvert$. As PFR only deducts integer values, we know that at the beginning of this iteration there must exist vertices in $G$ adjacent to a weight of at least $\ell - \lfloor p'\rfloor$ in $G$ exclusively connected to candidates in $C'$. In iteration $\lvert C'\rvert$ PFR adds a vertex connected to candidates in $C'$ with a total edge weight of $\lfloor p'\rfloor$. As we round according to $G$ we know that each of these vertices will have a degree equal to the weight of their outgoing edges, with each of these vertices being included in the committee. Hence, we will select at least $\ell - \lfloor\sum_{i \in N'}\sum_{c \in C'} p_c(i)\rfloor +\lfloor \sum_{i \in N'}\sum_{c \in C'} p_c(i)\rfloor = \ell$ candidates from $C'$, and therefore, the solid coalition $N'$ is satisfies according to PSC.
 }  
\end{proof}
Adding on to this, we however also show that PFR cannot be used to obtain the \new{recently introduced} stronger guarantee of rank-PJR \citep{BrPe23a}. \new{To define rank-PJR, for any $i \in N$ and $c \in C$ we call $\rank(i,c) \coloneqq \lvert \{c' \in C \colon c' \succ_i c\}\rvert + 1$ to be the rank given to $c$ by voter $i$ (i.e., voter $i$ gives candidate $c$ a rank of $r$ if this candidate appears in the $r$-th place of the voters rankings). Using this, we can define the property of rank-PJR.}
\begin{definition}[rank-PJR]
    A committee $W$ satisfies \emph{rank-proportional justified representation (rank-PJR)}\citep{BrPe23a} if there is no group $N'\subseteq N$ of voters, $C' \subseteq C$ of candidates, $\ell \in [k]$, and rank $r\in [m]$ such that 
    \begin{itemize}
        \item[i.)] $\rank(i,c) \le r$ for all $i \in N'$ and $c \in C'$
        \item[ii.)] $\lvert N'\rvert \ge \ell \frac{n}{k}$
        \item[iii.)] $\lvert \{c \in W \colon \rank(i,c) \le r \text{ for some } i \in N'\}\rvert < \ell.$
    \end{itemize}
\end{definition}
It is easy to see that any committee satisfying rank-PJR also satisfies PSC \citep{BrPe23a}. Further, using an example of \citeauthor{BrPe23a} we can show that PFR cannot be decomposed to satisfy ex-post rank-PJR.
\begin{example}
    Consider Example~5 of \citet{BrPe23a} with $4$ voters, $6$ candidates, and $k = 2$:
    \begin{itemize}
    \item $c_1 \succ c_2 \succ c_3 \succ c_4 \succ c_5 \succ c_6$ 
    \item $c_5 \succ c_2 \succ c_3 \succ c_4 \succ c_6 \succ c_1$ 
    \item $c_4 \succ c_3 \succ c_2 \succ c_6 \succ c_1 \succ c_5$ 
    \item $c_6 \succ c_3 \succ c_2 \succ c_5 \succ c_1 \succ c_4$.  \end{itemize}
    Here, PFR would place a probability mass of $\frac{1}{2}$ on each of $c_1, c_4, c_5, c_6$. This, however, is not compatible with rank-PJR. \new{Independent of which two voters get their first choice candidate in the deterministic committee, the other two voters will rank candidate $c_2$ at a rank of at least $3$ while receiving no candidate at a rank of $3$, thus constituting a rank-PJR violation.} \new{In fact, as the PFR outcome is the only ex-ante PSC outcome in this instance,} this also shows that ex-ante PSC and ex-post rank-PJR are incompatible.
\end{example}

\subsection{Consequences for Party List Elections}
A special case of multiwinner voting are so-called party list elections \citep{BGP+24a} in which a candidate can not only be elected once but multiple times. Of special interest is the apportionment problem. In it, voters give their vote for a \emph{single} party, which then is aggregated into seats for the parties. That is, there are $m$ parties, with party $i$ receiving a $p_i$ share of the votes and $h$ seats to be distributed among the parties. Just as for multiwinner voting, there are also ``proportionality'' axioms for apportionment. Most prominently, the \emph{quota} axiom states that party $i$ should receive either $\lfloor p_i h \rfloor$ or $\lceil p_i h \rceil$ seats, i.e., the number of seats for each party should be correctly rounded. Since allocation rules satisfying quota and other prominent desirable axioms do not exist \citep{BaYo01a}, recent work has studied the probabilistic allocation of seats \citep{Grim04a, GPP24a, JGU+24a}: the vote shares of the parties get probabilistically rounded in such a way that each party receives their exact quota in expectation and always their rounded quota ex-post. Here, we shortly note that we can apply our result on PFR to show that if one additionally knew the complete rankings of each voter over all parties, one can get a randomized apportionment scheme satisfying the desirable quota properties in addition to PSC (when interpreting this apportionment instance as a multiwinner voting instance). 
\begin{theorem}
    There is a randomized apportionment method satisfying ex-ante exact quota, ex-post quota and PSC. 
\end{theorem}
\begin{proof}
    For this, we note that PFR when run on a party-list instance, leads to probabilities exactly equal to the vote shares (as the candidates of a party never run out). Thus, \Cref{thm:main} immediately gives us that ex-ante exact quota and PSC are satisfied. Further, due to the dependent rounding works, we also get that no party would be either over or underrepresented thus also giving us ex-post quota.
\end{proof}
\begin{figure}
\centering
\begin{tikzpicture}[every plot/.append style={line width=2.2pt}, scale = 0.6]

\definecolor{crimson2143940}{RGB}{214,39,40}
\definecolor{darkgray176}{RGB}{176,176,176}
\definecolor{darkorange25512714}{RGB}{51,117,56}
\definecolor{forestgreen4416044}{RGB}{148,203,236}
\definecolor{mediumpurple148103189}{RGB}{194,106,119}
\definecolor{orchid227119194}{RGB}{227,119,194}
\definecolor{sienna1408675}{RGB}{220,205,125}
\definecolor{steelblue31119180}{RGB}{31,119,180}

\begin{axis}[
legend columns=2, 
legend cell align={left},
legend style={
  fill opacity=0.8,
  draw opacity=1,
  draw=none,
  text opacity=1,
  at={(0.5,1.35)},
  line width=3pt,
  anchor=north,
   /tikz/column 2/.style={
  	column sep=10pt,
  }, font=\LARGE
},
legend entries={$k = 3$, $k = 4$},
tick align=outside,
tick pos=left,
x grid style={darkgray176},
xlabel={Number of candidates},
xmin=3.7, xmax=14.3,
xtick style={color=black},
y grid style={darkgray176},
ylabel={Number of instances},
ymin=-0.749999973283545, ymax=163,
ytick style={color=black},every tick label/.append style={font=\LARGE}, 
label style={font=\LARGE},
xticklabels={$2$,$4$,$6$,$8$,$10$,$12$,$14$},
xtick={2,4,6,8,10,12,14}
]
\addlegendimage{darkorange25512714, mark = x}
\addlegendimage{sienna1408675, mark = o}

\addplot [semithick, darkorange25512714, mark = x, mark size=3, mark options={solid}]
table {%
4 36
5 77
6 148
7 159
8 90
9 36
10 10
11 3
13 1
};
\addplot [semithick, sienna1408675, mark = o, mark size=3, mark options={solid}]
table {%
5 35
6 58
7 126
8 116
9 83
10 61
11 22
12 7
13 5
14 1
};
\end{axis}

\end{tikzpicture}
\caption{Number of instances with a given number of candidates for $k = 3$ and $k = 4$.}
\label{fig:numcands}
\end{figure}

\section{Experiments}
We conduct some small experiments to measure how many committees our algorithm in \Cref{thm:main} would actually randomize over. For this, we follow \citet{McGr23a} and study their dataset of Scottish local government elections. This dataset consists of 1049 multiwinner voting instances with ranked preferences with on average 4800 voters, between 4 and 14 candidates and typically a $k$ of either $3$ or $4$.

See \Cref{fig:numcands} for an overview over the distribution of the number of instances for the different values of $k$ and $m$. As the original preferences are not complete, we use the dataset of completed preferences constructed by \citet{BBMP25a}. In their experimental work studying proportionality in this dataset, \citeauthor{BBMP25a} already noted that in most instances, PSC is satisfied by almost all committees. In our experiments, we implement the dependent rounding scheme of \citet{GKPS06a} and sample $50.000$ committees according to it for each instance. We are not aware of any computationally efficient way to extract these probabilities from the bipartite graph itself.
\begin{figure*}[h!]
    \begin{subfigure}{\textwidth}
 \begin{subfigure}{.5\textwidth}
 \centering
 % This file was created with tikzplotlib v0.10.1.
\begin{tikzpicture}[every plot/.append style={line width=2.2pt},scale = 0.5]
\definecolor{crimson2143940}{RGB}{46,37,133}
\definecolor{darkgray176}{RGB}{176,176,176}
\definecolor{darkorange25512714}{RGB}{51,117,56}
\definecolor{forestgreen4416044}{RGB}{148,203,236}
\definecolor{mediumpurple148103189}{RGB}{194,106,119}
\definecolor{orchid227119194}{RGB}{227,119,194}
\definecolor{sienna1408675}{RGB}{220,205,125}
\definecolor{steelblue31119180}{RGB}{31,119,180}
\begin{axis}[
legend columns=2, 
legend cell align={left},
legend style={
  fill opacity=0.8,
  draw opacity=1,
  draw=none,
  text opacity=1,
  at={(0.65,1.45)},
  line width=3pt,
  anchor=north,
   /tikz/column 2/.style={
  	column sep=10pt,
  }, font=\LARGE
},
legend entries={$100\%$, $(100\%-70\%]$, $(70\%-40\%]$,$(40\%-10\%]$, $(10\%-0\%]$},
tick align=outside,
tick pos=left,
x grid style={darkgray176},
xlabel={Number of candidates in instance},
xmin=4, xmax=10,
xtick style={color=black},
y grid style={darkgray176},
ylabel={Ratio of candidates with probability},
ymin=0, ymax=1,
ytick style={color=black},every tick label/.append style={font=\LARGE}, 
label style={font=\LARGE},
xticklabels={$4$, $5$, $6$, $7$, $8$, $9$, $10$},
xtick={4,5,6,7,8,9,10},
area style,
]
\addlegendimage{darkorange25512714}
\addlegendimage{forestgreen4416044}
\addlegendimage{crimson2143940}
\addlegendimage{mediumpurple148103189}
\addlegendimage{sienna1408675}

% ≤10% (sienna color)
\addplot [fill=sienna1408675] coordinates
{(4,0.0)(5,0.023376623376623377)(6,0.07207207207207207)(7,0.14375561545372867)(8,0.22916666666666666)(9,0.2592592592592593)(10,0.31)
(10,0) (9,0) (8,0) (7,0) (6,0) (5,0) (4,0)};

% ≥50% (purple color)
\addplot [fill=mediumpurple148103189] coordinates
{(4,0.13888888888888884)(5,0.2987012987012987)(6,0.40990990990990994)(7,0.49955076370170715)(8,0.5819444444444444)(9,0.6203703703703705)(10,0.7)
(10,0.31)(9,0.2592592592592593)(8,0.22916666666666666)(7,0.14375561545372867)(6,0.07207207207207207)(5,0.023376623376623377)(4,0.0)};

% ≥80% (crimson color)
\addplot [fill=crimson2143940] coordinates
{(4,0.375)(5,0.6103896103896104)(6,0.7195945945945945)(7,0.788858939802336)(8,0.825)(9,0.8703703703703703)(10,0.88)
(10,0.7)(9,0.6203703703703705)(8,0.5819444444444444)(7,0.49955076370170715)(6,0.40990990990990994)(5,0.2987012987012987)(4,0.13888888888888884)};

% ≥90% (green color)
\addplot [fill=forestgreen4416044] coordinates
{(4,0.6666666666666667)(5,0.8545454545454545)(6,0.9346846846846847)(7,0.958670260557053)(8,0.9694444444444444)(9,0.9753086419753086)(10,0.97)
(10,0.88)(9,0.8703703703703703)(8,0.825)(7,0.788858939802336)(6,0.7195945945945945)(5,0.6103896103896104)(4,0.375)};
 
% 100% (orange color)
\addplot [fill=darkorange25512714] coordinates
{(4,0.6666666666666667)(5,0.8545454545454545)(6,0.9346846846846847)(7,0.958670260557053)(8,0.9694444444444444)(9,0.9753086419753086)(10,0.97)
(10,1) (9,1) (8,1) (7,1) (6,1) (5,1) (4,1)};
\end{axis}
\end{tikzpicture}
\end{subfigure}%
\hfil
\begin{subfigure}{.5\textwidth}
\centering
  % This file was created with tikzplotlib v0.10.1.
\begin{tikzpicture}[every plot/.append style={line width=2.2pt},scale = 0.5]
\definecolor{crimson2143940}{RGB}{46,37,133}
\definecolor{darkgray176}{RGB}{176,176,176}
\definecolor{darkorange25512714}{RGB}{51,117,56}
\definecolor{forestgreen4416044}{RGB}{148,203,236}
\definecolor{mediumpurple148103189}{RGB}{194,106,119}
\definecolor{orchid227119194}{RGB}{227,119,194}
\definecolor{sienna1408675}{RGB}{220,205,125}
\definecolor{steelblue31119180}{RGB}{31,119,180}
\begin{axis}[
legend columns=2, 
legend cell align={left},
legend style={
  fill opacity=0.8,
  draw opacity=1,
  draw=none,
  text opacity=1,
  at={(0.65,1.45)},
  line width=3pt,
  anchor=north,
   /tikz/column 2/.style={
  	column sep=10pt,
  }, font=\LARGE
},
legend entries={$100\%$, $(100\%-70\%]$, $(70\%-40\%]$,$(40\%-10\%]$, $(10\%-0\%]$},
tick align=outside,
tick pos=left,
x grid style={darkgray176},
xlabel={Number of candidates in instance},
xmin=4, xmax=10,
xtick style={color=black},
y grid style={darkgray176},
ylabel={Ratio of candidates with probability},
ymin=0, ymax=1,
ytick style={color=black},every tick label/.append style={font=\LARGE}, 
label style={font=\LARGE},
xticklabels={$5$, $6$, $7$, $8$, $9$, $10$, $11$},
xtick={4,5,6,7,8,9,10},
area style,
]
\addlegendimage{darkorange25512714}
\addlegendimage{forestgreen4416044}
\addlegendimage{crimson2143940}
\addlegendimage{mediumpurple148103189}
\addlegendimage{sienna1408675}

% ≤10% (sienna color)
\addplot [fill=sienna1408675] coordinates
{(4,0.0) (5,0.017241) (6,0.066893) (7,0.117457) (8,0.182062) (9,0.216393) (10,0.280992)
(10,0) (9,0) (8,0) (7,0) (6,0) (5,0) (4,0)};

% ≥50% (purple color)
\addplot [fill=mediumpurple148103189] coordinates
{(4,0.0) (5,0.017241) (6,0.066893) (7,0.117457) (8,0.182062) (9,0.216393) (10,0.280992)
(10,0.6446280991735538)(9,0.5475409836065575)(8,0.49665327978580986)(7,0.4224137931034483)(6,0.3310657596371882)(5,0.1954022988505747)(4,0.07428571428571418)};

% ≥80% (crimson color)
\addplot [fill=crimson2143940] coordinates
{(4,0.07428571428571418)(5,0.1954022988505747)(6,0.3310657596371882)(7,0.4224137931034483)(8,0.49665327978580986)(9,0.5475409836065575)(10,0.6446280991735538)
(10,0.8099173553719008)(9,0.8065573770491803)(8,0.7563587684069611)(7,0.697198275862069)(6,0.6179138321995465)(5,0.4885057471264368)(4,0.3028571428571428)};

% ≥90% (green color)
\addplot [fill=forestgreen4416044] coordinates
{(4,0.3028571428571428)(5,0.4885057471264368)(6,0.6179138321995465)(7,0.697198275862069)(8,0.7563587684069611)(9,0.8065573770491803)(10,0.8099173553719008)
(10,0.9256198347107438)(9,0.9377049180327869)(8,0.9116465863453815)(7,0.9170258620689655)(6,0.8673469387755102)(5,0.8333333333333334)(4,0.6171428571428571)};

% 100% (orange color)
\addplot [fill=darkorange25512714] coordinates
{(4,0.6171428571428571)(5,0.8333333333333334)(6,0.8673469387755102)(7,0.9170258620689655)(8,0.9116465863453815)(9,0.9377049180327869)(10,0.9256198347107438)
(10,1) (9,1) (8,1) (7,1) (6,1) (5,1) (4,1)};
\end{axis}
\end{tikzpicture}
  \label{fig:candprobs}
  
\end{subfigure}%
\end{subfigure}
\caption{Both plots depict the fraction of candidates who are assigned are assigned a probability in a certain range  as a stackplot (for $k = 3$ on the left and $k = 4$ on the right). }
\label{fig:distr}
\end{figure*}
First, for the individual probabilities given by PFR, we see in ~\Cref{fig:distr} that typically --- especially for larger values of $m$ --- about $20\%$ to $40 \%$ of the candidates have a probability of at least $70\%$ to get chosen by PFR. \new{On the other hand, the ratio of candidates with a probability of less than $40\%$, is increasing with $m$ and even reaches $50\%$ for larger values of $m$. Thus, several candidates, typically have quite a low, yet non neglible, probability of getting elected.} 

Second, we measure the total number of committees with a positive probability and then the number of committees in the top $75\%$ (i.e., how many committees are needed to get a probability mass of at least $75\%$, if there were four committees sampled, each with probability $25\%$ this number would be $3$), $90\%$, $95\%$, and $99\%$ and compare these to the total number of committees satisfying PSC. We display these in \Cref{fig:sub1} for all values of $k$ and $m$ with more than $10$ instances. In it, we can see the effect already noticed by \citet{BBMP25a}, most committees, especially for smaller values of $m$ satisfy PSC. Indeed, also most of these committees satisfying PSC were also sampled by us after fifty thousand samples. In fact, as already observed by \citet{BBMP25a} in several instances, barely any non-trivial solid coalitions exist. Thus, there are very few constraints on the output of PFR, which in turn leads to most committees being able to be selected.  For the number of committees sampled with higher probabilities, one can for instance see that on average for $k = 3$ and $m = 7$ one of $7$ committees will be sampled with probability at least $75\%$ while on average with $99\%$ probability it will be one of $19$, with on average there being about $25$ out of the $35$ total committees satisfying PSC. \new{We further note that the committee with the largest probability of being selected usually has a probability between $15\%$ and $35\%$.}
\wrapstuffclear
\begin{figure*}[t]
  \begin{subfigure}{.5\textwidth}
  \centering
  % This file was created with tikzplotlib v0.10.1.
\begin{tikzpicture}[every plot/.append style={line width=2.2pt}, scale = 0.6]

\definecolor{crimson2143940}{RGB}{46,37,133}
\definecolor{darkgray176}{RGB}{176,176,176}
\definecolor{darkorange25512714}{RGB}{51,117,56}
\definecolor{forestgreen4416044}{RGB}{148,203,236}
\definecolor{mediumpurple148103189}{RGB}{194,106,119}
\definecolor{orchid227119194}{RGB}{227,119,194}
\definecolor{sienna1408675}{RGB}{220,205,125}
\definecolor{steelblue31119180}{RGB}{31,119,180}

\begin{axis}[
legend columns=3, 
legend cell align={left},
legend style={
  fill opacity=0.8,
  draw opacity=1,
  draw=none,
  text opacity=1,
  at={(0.5,1.35)},
  line width=3pt,
  anchor=north,
   /tikz/column 2/.style={
  	column sep=10pt,
  }, font=\LARGE
},
legend entries={PSC,
	Positive,
	$99\%$, $95\%$,
	$90\%$, $75\%$},
tick align=outside,
tick pos=left,
x grid style={darkgray176},
xlabel={Number of candidates},
xmin=3.7, xmax=10.3,
xtick style={color=black},
y grid style={darkgray176},
ylabel={Average number of committees},
ymin=-0.749999973283545, ymax=80,
ytick style={color=black},every tick label/.append style={font=\LARGE}, 
label style={font=\LARGE},
xticklabels={$4$,$5$,$6$,$7$,$8$,$9$, $10$,$11$,$14$},
xtick={4,5,6,7,8,9,10,11}
]
\addlegendimage{darkorange25512714, mark = x, mark size=5}
\addlegendimage{forestgreen4416044, mark = o, mark size=4}
\addlegendimage{orchid227119194, mark = diamond, mark size=4}
\addlegendimage{sienna1408675, mark = square, mark size=4}
\addlegendimage{mediumpurple148103189, mark = oplus, mark size=4}
\addlegendimage{crimson2143940, mark = asterisk, mark size=5}

\addplot [semithick, darkorange25512714, mark = x, mark size=3, mark options={solid}]
table {%
4 2.86111111111111
5 6.83116883116883
6 15.0067567567568
7 25.7295597484277
8 40.7333333333333
9 63.0833333333333
10 78
};
\addplot [semithick, forestgreen4416044, , mark = o, mark size=3, mark options={solid}]
table {%
4 2.66666666666667
5 6.72727272727273
6 14.3040540540541
7 25.0188679245283
8 39.4
9 60
10 75.7
};
\addplot [semithick, crimson2143940, , mark = asterisk, mark size=4, mark options={solid}]
table {%
4 1.86111111111111
5 3.11688311688312
6 4.5
7 6.57232704402516
8 8.44444444444444
9 10.8055555555556
10 15.5
};
\addplot [semithick, mediumpurple148103189, , mark = oplus, mark size=3, mark options={solid}]
table {%
4 2.38888888888889
5 4.54545454545455
6 7.2027027027027
7 10.9811320754717
8 14.9333333333333
9 20.25
10 29.1
};
\addplot [semithick, sienna1408675, , mark = square, mark size=3, mark options={solid}]
table {%
4 2.5
5 5.22077922077922
6 8.95945945945946
7 14.0251572327044
8 19.6666666666667
9 26.9166666666667
10 39
};
\addplot [semithick, orchid227119194, , mark = diamond, mark size=3, mark options={solid}]
table {%
4 2.66666666666667
5 6.31168831168831
6 11.7432432432432
7 19.5031446540881
8 28.5666666666667
9 40.4444444444444
10 55.9
};

\addplot [line width=1.5pt, black,dashed]
table{
4 4
5 10
6 20
7 35
8 56
9 84
10 120
};
\end{axis}

\end{tikzpicture}
\end{subfigure}%
 \begin{subfigure}{.5\textwidth}
  \centering
  % This file was created with tikzplotlib v0.10.1.
\begin{tikzpicture}[every plot/.append style={line width=2.2pt},scale = 0.6]

\definecolor{crimson2143940}{RGB}{46,37,133}
\definecolor{darkgray176}{RGB}{176,176,176}
\definecolor{darkorange25512714}{RGB}{51,117,56}
\definecolor{forestgreen4416044}{RGB}{148,203,236}
\definecolor{mediumpurple148103189}{RGB}{194,106,119}
\definecolor{orchid227119194}{RGB}{227,119,194}
\definecolor{sienna1408675}{RGB}{220,205,125}
\definecolor{steelblue31119180}{RGB}{31,119,180}

\begin{axis}[
legend columns=3, 
legend cell align={left},
legend style={
  fill opacity=0.8,
  draw opacity=1,
  draw=none,
  text opacity=1,
  at={(0.5,1.35)},
  line width=3pt,
  anchor=north,
   /tikz/column 2/.style={
  	column sep=10pt,
  }, font=\LARGE
},
legend entries={PSC,
	Positive,
	$99\%$, $95\%$,
	$90\%$, $75\%$},
tick align=outside,
tick pos=left,
x grid style={darkgray176},
xlabel={Number of candidates},
xmin=4.7, xmax=11.3,
xtick style={color=black},
y grid style={darkgray176},
ylabel={Average number of committees},
ymin=-0.749999973283545, ymax=150,
ytick style={color=black},every tick label/.append style={font=\LARGE}, 
label style={font=\LARGE},
xticklabels={$5$, $6$, $7$, $8$, $9$, $10$, $11$},
xtick={5,6,7,8,9,10,11}
]
\addlegendimage{darkorange25512714, mark = x, mark size=5}
\addlegendimage{forestgreen4416044, mark = o, mark size=4}
\addlegendimage{orchid227119194, mark = diamond, mark size=4}
\addlegendimage{sienna1408675, mark = square, mark size=4}
\addlegendimage{mediumpurple148103189, mark = oplus, mark size=4}
\addlegendimage{crimson2143940, mark = asterisk, mark size=5}

\addplot [semithick, darkorange25512714, mark = x, mark size=3, mark options={solid}]
table {%
5 3.51428571428571
6 10.1034482758621
7 20.031746031746
8 40.5862068965517
9 62.3734939759036
10 113.72131147541
11 126.363636363636
};
\addplot [semithick, forestgreen4416044, mark = o, mark size=3, mark options={solid}]
table {%
5 2.94285714285714
6 9.17241379310345
7 19.0555555555556
8 38.3706896551724
9 59.433734939759
10 105.049180327869
11 115.363636363636
};
\addplot [semithick, crimson2143940, mark = asterisk, mark size=4, mark options={solid}]
table {%
5 2.08571428571429
6 3.72413793103448
7 5.54761904761905
8 8.17241379310345
9 11.1084337349398
10 17.1639344262295
11 20.6818181818182
};
\addplot [semithick, mediumpurple148103189, mark = oplus, mark size=3, mark options={solid}]
table {%
5 2.51428571428571
6 5.55172413793103
7 9.03174603174603
8 14.2931034482759
9 20.3012048192771
10 32.5737704918033
11 39.8636363636364
};
\addplot [semithick, sienna1408675, mark = square, mark size=3, mark options={solid}]
table {%
5 2.77142857142857
6 6.74137931034483
7 11.2936507936508
8 18.7068965517241
9 27.3012048192771
10 44.3606557377049
11 54.2727272727273
};
\addplot [semithick, orchid227119194, mark = diamond, mark size=3, mark options={solid}]
table {%
5 2.91428571428571
6 8.46551724137931
7 15.2619047619048
8 27.4310344827586
9 41.0361445783133
10 69.0655737704918
11 81.0909090909091
};
\addplot [line width=1.5pt, black,dashed]
table{
5 5
6 15
7 35
8 70
9 126
10 210
11 330
};
\end{axis}

\end{tikzpicture}
\end{subfigure}%
\caption{Average number of committees satisfying PSC, having a positive probability of being sampled after 50 thousand samples, and how many committees are needed to get a probability mass of at least $75\%$, $90\%$, $95\%$, and $99\%$, for $k = 3$ on the left and $k = 4$ on the right. The black dashed line indicates the total number of committees.}
\label{fig:sub1}
\end{figure*}

\new{Thus, while PFR offers an appealing theoretical solution, it also inherits several key weaknesses of PSC as an axiom. Namely, due to its dependence on solid coalitions, the relative lack of meaningful solid coalitions} \citep{BBMP25a} \new{in-turn leads to a lack of constraints for PFR, which still enables a large fraction of all PSC committees to be chosen by PFR. Hence, if one would want near-certainty on the outcome of the lottery, PFR is most likely not a good choice as an election rule. If, however, one is content with its use as a lottery, PFR might provide an interesting monotone alternative to classical deterministic voting rules.}  
\section{Conclusion}
We studied the problem of finding multiwinner voting rules that are both candidate monotone and satisfy the proportionality axiom of \emph{proportionality for solid coalitions}. First, we show that the \emph{non-resolute} rule selecting all committees satisfying PSC is indeed candidate monotone. Secondly, we design a \emph{probabilistic} voting rule based on \emph{Phragmén's Ordered Rule} and show that this probabilistic voting rule is both candidate monotone and can be implemented to randomize only over committees satisfying PSC. 

There are several open questions and directions for future work. Most importantly, whether there is a \emph{resolute} voting rule satisfying PSC and candidate monotonicity is still open. Further, it is open whether the results presented here can be extended to stronger proportionality axioms (such as rank-PJR of \citet{BrPe23a}) or to the setting with weak preferences and the axiom of generalized PSC. For both, a characterization as provided by \citet{AzLe22a} of the respective proportionality axioms is missing as well.

\section{Acknowledgments}
Thank you to Dominik Peters for various discussions on this topic. This work was partially supported by the Singapore Ministry of Education under grant number MOE-T2EP20221-0001.

\section{Compliance with Ethical Standards}
This work was partially supported by the Singapore Ministry of Education under grant number MOE-T2EP20221-0001. No human participants or animals were used for this paper.
\bibliographystyle{abbrvnat}
\bibliography{dist, abb, algo}
\end{document}